\documentclass[11pt,a4paper]{article}
\usepackage[left=25mm, top=25mm,bottom=25mm]{geometry}
\usepackage{amsmath}
\usepackage{amsthm}
\usepackage{ mathrsfs }
\usepackage{amssymb}

\newtheorem{theorem}{Theorem}

\newtheorem{remark}[theorem]{Remark}

\numberwithin{equation}{section}
\numberwithin{theorem}{section}

\DeclareMathOperator{\tr}{tr}

\begin{document}

\title{A Variational Approach to Lax Representations}

\author{Duncan Sleigh, Frank Nijhoff and Vincent Caudrelier\\ 
School of Mathematis, University of Leeds}

\maketitle

\abstract{It is shown that the Zakharov-Mihailov (ZM) Lagrangian structure for integrable nonlinear equations derived from a general class of Lax pairs possesses a Lagrangian multiform structure in the sense of \cite{Lobb2009}.  We show that, as a consequence of this multiform structure, we can formulate a variational principle for the Lax pair itself, a problem that to our knowledge was never previously considered.  As an example, we present an integrable $N\times N$ matrix system that contains the AKNS hierarchy, and we exhibit the Lagrangian multiform structure of the scalar AKNS hierarchy by presenting the components corresponding to the first three flows of the hierarchy.  }

\section{Introduction}

A unifying principle in physics is that at the classical level the fundamental theories of interactions can be described by Lagrangian field 
theories through the least-action principle applied to the corresponding action functional. There are strong indications that field theories that 
are classically integrable in the sense of being solvable through the inverse scattering method possess a Lagrangian description. In fact, Zakharov and Mikhailov 
showed in \cite{Zakharov1980} that integrable systems derived from a fairly broad class of Lax pairs, i.e. based on zero-curvature conditions, in 1+1 
dimensions, possess a natural Lagrangian description, while the Lax pair allows for the application of the inverse scattering method or other dressing 
method techniques for their exact solution. A key aspect of such integrable systems is the notion of multidimensional consistency (MDC), namely the 
fact that the defining equations can be seen as members of compatible hierarchies of equations in terms of an, in principle, arbitrary number of independent variables, 
which can be simultaneously and consistently imposed on one and the same set of dependent variables (or the same set of components of the independent 
variable in the multi-component case). Alternatively this can be interpreted as the existence of an infinite hierarchy of symmetries for those 
constitutive equations. A point of view is that the true integrable system in question is the collection of all these compatible equations, i.e. we can consider the 
entire multidimensionally consistent system as the integrable system.\\

\noindent From that latter point of view, a key weakness of the conventional Lagrangian description is that it fails to capture the multidimensional consistency. After all, a scalar 
Lagrangian function will only provide one single equation of the motion per component of the system, and not the whole system of compatible equations for those 
components. Thus, there is a mismatch between the integrability and the Lagrangian aspect of the variational approach, which we would like to be able to 
encode all relevant aspects of the system including its multidimensional consistency. This weakness was overcome in the paper \cite{Lobb2009} where it was proposed 
to extend the classical scalar Lagrangian (or volume form with respect to the space of independent variables) to a genuine differential $d$-form in the space of 
independent variables of $N$ dimensions with $d<N$, where $d$ corresponds to the dimensionality of the equations (i.e. $d=2$ for systems of  
1+1-dimensional equations). This led to the introduction of a new notion of a {\em Lagrangian multiform} \cite{Lobb2009,Xenitidis3295}, where the multidimensional consistency manifests itself by the Euler-Lagrange (EL) equations of the Lagrangian $d$-form being independent of the choice of the surface of integration in the action functional. 

\section{Multiforms and Lax Equations}\label{s:MLE}
\subsection{Lagrangian 2-Forms}

To make these ideas more precise, we define a Lagrangian 2-form in a space of $N>2$ independent variables $\xi_1,\ldots,\xi_N$ by an expression 
\begin{equation}\label{eq:2form}
\sf{L}:= \sum_{i<j} \mathscr{L}_{(ij)} {\rm d}\xi_i\wedge {\rm d}\xi_j 
\end{equation} 
where the components $\mathscr{L}_{(ij)}$ are functions of (scalar or multi-component) field variables $\boldsymbol{\varphi}(\xi_1,\dots,\xi_N)$ and their derivatives. 
The action functional, defined by 
\begin{equation}\label{eq:action}
\mathscr{S}[\boldsymbol{\varphi}(\xi_1,\dots,\xi_N);\sigma]=\int_\sigma\,\sf{L}\  , 
\end{equation}
should be considered as a functional not only of the field variables $\boldsymbol{\varphi}$, but also of the two-dimensional surface $\sigma$ in $N$-dimensional space of 
independent variables. Thus, the relevant variational principle requires that the action $\mathscr{S}$ is stationary with respect to infinitesimal 
variations of the fields $\boldsymbol{\varphi}\to\boldsymbol{\varphi}+\delta\boldsymbol{\varphi}$ for every choice of the surface $\sigma$. In the language of the variational bicomplex (see e.g. \cite{Dickey2003}) this leads to the relation 

\begin{equation}\label{deltaL}
\delta \rm d\sf{L}=0
\end{equation}
where

\begin{equation} 
\rm d\textsf{L}=\sum_{i\neq j,k}D_i\mathscr{L}_{(jk)}{\rm d}\xi_i\wedge {\rm d}\xi_j \wedge {\rm d}\xi_k , 
\end{equation}

\begin{equation}
D_i:=\frac{\partial}{\partial \xi_i}+\sum_I\boldsymbol{\varphi}_{I\xi_i}\frac{\partial}{\partial \boldsymbol{\varphi}_I}
\end{equation}
where $I=(j_1,\ldots,j_N)$ and

\begin{equation}
\boldsymbol{\varphi}_I:=\frac{\partial ^{|I|}\boldsymbol{\varphi}}{(\partial \xi_1)^{j_1}\ldots(\partial \xi_N)^{j_N}}
\end{equation}
with $|I|= j_1+\ldots +j_N$ and $I\xi_k=(j_1,\ldots,j_{k+1},\ldots,j_N)$.  Then

\begin{equation}
\delta \rm d\textsf{L}:=\sum_I \frac{\partial \rm d\textsf{L}}{\partial \boldsymbol{\varphi}_I}\delta \boldsymbol{\varphi}_I ,
\end{equation}
The proof of \eqref{deltaL} (originally given in \cite{Suris2016}) can be found in the first part of Appendix \ref{proof1}.  For an autonomous $\sf{L}$, i.e. one that has no direct dependence on the independent variables $\xi_i$, \eqref{deltaL} implies that $\textrm{d}\textsf{L}=c$, a constant.  This must hold, not as an identity for arbitrary $\boldsymbol{\varphi}$, but on solutions of the Euler-Lagrange (EL) equations for those fields.

\begin{remark}
In \cite{Lobb2009} the slightly stronger ``closure'' relation $\textrm{d}\textsf{L}=0$ was proposed as a condition of stationarity of the action with respect to variations of the action as a functional of the surface $\sigma$. To date, all known examples of Lagrangian multiforms beyond the 1-form case obey the latter condition.
\end{remark}

\noindent A direct consequence of the relation \eqref{deltaL} is that there are constraints on the components $\mathscr{L}_{(ij)}$ that constitute the multiform $\sf{L}$. We shall refer to these constraints as the \textbf{multiform EL equations} of the Lagrangian multiform $\sf{L}$.  We introduce the notation
\begin{equation}
\boldsymbol{\varphi}_{i^aj^bk^c}:=\bigg(\frac{\partial}{\partial \xi_i}\bigg)^a\bigg(\frac{\partial}{\partial \xi_j}\bigg)^b\bigg(\frac{\partial}{\partial \xi_k}\bigg)^c\boldsymbol{\varphi} \quad\text{for } a,b,c \in \mathbb{Z}_{\geq 0}
\end {equation}
and define the variational derivative

\begin{equation}
\frac{\delta\mathscr{L}_{(ij)}}{\delta \boldsymbol{\varphi}_{i^aj^bk^c}}=\sum_{\alpha ,\beta \geq 0}(-1)^{\alpha + \beta} D_i^\alpha D_j^\beta\frac{\partial \mathscr{L}_{(ij)}}{\partial \boldsymbol{\varphi}_{i^{a+\alpha}j^{b+\beta}k^c}} \quad \text{for } a,b,c\in \mathbb{Z}_{\geq 0}
\end{equation}
with $\dfrac{\delta\mathscr{L}_{(ij)}}{\delta \boldsymbol{\varphi}_{i^aj^bk^c}}:=0$ in the case where one or more of $a,b,c$ is negative.  We shall use the convention that

\begin{equation}
\bigg(\frac{\delta \mathscr{L}}{\delta \boldsymbol{\varphi}}\bigg)_{ij}=\frac{\delta \mathscr{L}}{\delta \boldsymbol{\varphi}_{ji}}
\end{equation}
when taking variational derivatives with respect to the matrix-valued field $\boldsymbol{\varphi}$.

\begin{theorem}\label{thm:ELeq}

 The multiform EL equations for a Lagrangian 2-form are given by

\begin{equation}\label{eq:ELeqs}
\frac{\delta\mathscr{L}_{(ij)}}{\delta \boldsymbol{\varphi}_{i^lj^mk^{n-1}}}+\frac{\delta\mathscr{L}_{(jk)}}{\delta \boldsymbol{\varphi}_{i^{l-1}j^mk^n}}+\frac{\delta\mathscr{L}_{(ki)}}{\delta \boldsymbol{\varphi}_{i^lj^{m-1}k^n}}=0
\end{equation}
for $l,m,n\geq 0 $. 
\end{theorem}
\noindent This is equivalent to the equations given in \cite{Suris2016a,Suris2016} that were derived by approximating the surface of integration, $\sigma$, by a stepped surface. We give an alternative proof of this in Appendix \ref{proof1}, that does not require the use of stepped surfaces.

\begin{remark}
The multiform EL equations include the \textbf{standard EL equations} and the \textbf{higher jet EL equations}.  For example, in the case where $l=0,m=0,n=1$, \eqref{eq:ELeqs} gives us 

\begin{equation}
\frac{\delta\mathscr{L}_{(ij)}}{\delta \boldsymbol{\varphi}}=0,
\end{equation}
i.e. the standard EL equations for $\mathscr{L}_{(ij)}$.  In the case where more than one of $l,m$ and $n$ is greater than zero, we get the higher jet EL equations.
\end{remark}

\noindent In the present paper we apply the idea of Lagrangian multiform to the Lagrangian density proposed in \cite{Zakharov1980}, and we show that it can indeed be extended 
naturally to a Lagrangian 2-form structure. This makes the multidimensional consistency of the corresponding Zakharov-Mikhailov (ZM) system manifest at the Lagrangian 
level. Furthermore we show that, in a precise sense, our Lagrangian multiform leads to a variational formulation of the underlying Lax pair itself. In fact, 
the 2-form structure leads naturally to the Lagrangian description for a {\em Lax triplet} (or more generally a {\em Lax multiplet}), and thus we can recover the Lax pair from the Lagrangian multiforms associated 
with the ZM Lagrangians. 

\subsection{The Zakharov-Mikhailov Lagrangian}\label{s:ZML}

Following the method of Zakharov and Mikhailov \cite{Zakharov1980} we start from a $N\times N$ matrix Lax pair $U$ and $V$ and auxiliary problem
\begin{equation}
\label{eq:Lax}
\Psi_\xi=U(\xi,\eta,\lambda)\Psi, \quad \Psi_\eta=V(\xi,\eta,\lambda)\Psi.
\end{equation}
Henceforth, we shall commit an abuse of terminology and refer to the $N\times N$ matrix $\Psi$ as the eigenfunction of the Lax pair. This gives rise to the compatibility condition
\begin{equation}\label{eq:ZS} 
U_\eta-V_\xi+[U,V]=0.
\end{equation}
We assume that $U$ and $V$ are rational functions of $\lambda$ with a finite number of distinct simple poles (the case where $U$ and $V$ have higher order poles is dealt with in \cite{Dickey2003}), so
\begin{equation}\label{eq:UV} 
U=U^0(\xi,\eta)+\sum_{i=1}^{N_1}\frac{U^i(\xi,\eta)}{\lambda-a_i},\quad V=V^0(\xi,\eta)+\sum_{j=1}^{N_2}\frac{V^j(\xi,\eta)}{\lambda-b_j}.
\end{equation}

\noindent giving the compatibility conditions

\begin{equation}\label{comp1}
U_\eta^0-V_\xi^0+[U^0,V^0]=0
\end{equation}
and
\begin{equation}\label{comp2}
U_\eta^i+\bigg[U^i,V^0+\sum_{j=1}^{N_2}\frac{V^j}{a_i-b_j}\bigg]=0,\quad V_\xi^j+\bigg[V^j,U^0+\sum_{i=1}^{N_1}\frac{U^i}{b_j-a_i}\bigg]=0.
\end{equation}

\noindent Equation (\ref{comp1}) implies that $U^0$ and $V^0$ can be written in terms of an invertible matrix $g(\xi,\eta)$ such that

\begin{equation}\label{uzero}
U^0=g_\xi g^{-1}, \quad V_0=g_\eta g^{-1}.
\end{equation}

\noindent The matrices $U^i$ and $V^j$ are expressed as

\begin{equation}\label{bardef}
U^i=\varphi^i\bar U^i(\varphi^i)^{-1}, \quad V^j=\psi^j\bar V^j(\psi^j)^{-1}
\end{equation}

\noindent where $\bar U^i$ and $\bar V^j$ are the Jordan normal forms of $U^i$ and $V^j$ which depend only on $\xi$ and $\eta$ respectively.  In order to show that $\bar U^i$ depends only on $\xi$, we let $Y^i$ be the solution of

\begin{equation}
Y^i_\eta=V|_{\lambda=a_i}Y^i
\end{equation}

\noindent then

\begin{equation}\label{eq:Yeq}
\begin{split}
\partial_\eta((Y^i)^{-1}U^iY^i)=-(Y^i)^{-1}V|_{\lambda=a_i}U^iY^i+(Y^i)^{-1}[V|_{\lambda=a_i},U^i]Y^i+(Y^i)^{-1}U^iV|_{\lambda=a_i}Y^i=0,
\end{split}
\end{equation}

\noindent so $(Y^i)^{-1}U^iY^i$ is constant with respect to $\eta$.  Since similarity transformations preserve eigenvalues, this tells us that the eigenvalues of $U^i$ do not depend on $\eta$.  Therefore the Jordan normal matrix $\bar U^i$ which has the same eigenvalues as $U^i$ does not depend on $\eta$.  Similarly $\bar V^j$ does not depend on $\xi$.\\

\noindent The ZM Lagrangian

\begin{equation}\label{ZMlag}
\begin{split}
\mathscr{L}_{(\xi\eta)}=\tr\bigg\{\sum_{i=1}^{N_1}(\varphi^i)^{-1}(\varphi_\eta^i-g_\eta g^{-1}\varphi^i)\bar U^i-\sum_{j=1}^{N_2}(\psi^j)^{-1}(\psi_\xi^j-g_\xi g^{-1}\psi^j)\bar V^j\\-\sum_{i=1}^{N_1}\sum_{j=1}^{N_2}\frac{\psi^j\bar V^j(\psi^j)^{-1}\varphi^i\bar U^i(\varphi^i)^{-1}}{a_i-b_j}\bigg\}
\end{split}
\end{equation}

\noindent has EL equations equivalent to the compatibility conditions (\ref{comp2}).  We find that
\begin{subequations}
\begin{equation}\label{varphivarderv}
\begin{split}
\frac{\delta \mathscr{L}}{\delta \varphi^i}=&-(\varphi^i)^{-1}\bigg(\varphi^i_\eta-(g_\eta g^{-1}+\sum_{j=1}^{N_2}\frac{\psi^j \bar V^j (\psi^j)^{-1}}{a_i-b_j})\varphi^i\bigg)\bar U^i(\varphi^i)^{-1}\\&+\bar U^i(\varphi^i)^{-1}\bigg(\varphi^i_\eta-(g_\eta g^{-1}+\sum_{j=1}^{N_2}\frac{\psi^j \bar V^j (\psi^j)^{-1}}{a_i-b_j})\varphi^i\bigg)(\varphi^i)^{-1}
\end{split}
\end{equation}
and

\begin{equation}\label{psivarderv}
\begin{split}
\frac{\delta \mathscr{L}}{\delta \psi^j}=&(\psi^j)^{-1}\bigg(\psi^j_\xi-(g_\xi g^{-1}+\sum_{i=1}^{N_1}\frac{\varphi^i \bar U^i (\varphi^i)^{-1}}{b_j-a_i})\psi^j\bigg)\bar V^j(\psi^j)^{-1}\\&-\bar V^j(\psi^j)^{-1}\bigg(\psi^j_\xi-(g_\xi g^{-1}+\sum_{i=1}^{N_1}\frac{\varphi^i \bar U^i (\varphi^i)^{-1}}{b_j-a_i})\psi^j\bigg)(\psi^j)^{-1}
\end{split}
\end{equation}
\end{subequations}
which, when we use \eqref{bardef} and set equal to zero are equivalent to \eqref{comp2}.

\begin{remark}
From the definition of $\varphi^i$ and $\psi^j$ in \eqref{bardef}, it is clear that they are not-unique.  As a result, \eqref{varphivarderv} is equivalent to the statement that

\begin{equation}
(\varphi^i)^{-1}(\varphi^i_\eta-(g_\eta g^{-1}+\sum_{j=1}^{N_2}\frac{\psi^j \bar V^j (\psi^j)^{-1}}{a_i-b_j})\varphi^i)
\end{equation}
can be any matrix that commutes with $\bar U^i$.  A similar statement relating to $\bar V^j$ follows from \eqref{psivarderv}.  However, the non-uniqueness of $\varphi^i$ and $\psi^j$ does not lead to any additional freedom on solutions of the system because, by \eqref{bardef}, this freedom does not affect $U^i$ and $V^j$.
\end{remark}
\noindent We also find that

\begin{equation}
\begin{split}
\frac{\delta \mathscr{L}}{\delta g}=& \sum_{i=1}^{N_1}\bigg\{D_\eta(g^{-1}\varphi^i\bar U^i(\varphi^i)^{-1})+g^{-1}\varphi^i\bar U^i(\varphi^i)^{-1}g^{-1}g_\eta\bigg\}\\&-\sum_{j=1}^{N_2}\bigg\{D_\xi(g^{-1}\psi^j\bar V^j(\psi^j)^{-1})+g^{-1}\psi^j\bar V^j(\psi^j)^{-1}g^{-1}g_\xi\bigg\}.
\end{split}
\end{equation}
When we use \eqref{bardef} and set equal to zero this is equivalent to

\begin{equation}\label{eq:varg}
\sum_{i=1}^{N_1}\bigg\{U^i_\eta+[U^i,V^0]\bigg\}=\sum_{j=1}^{N_2}\bigg\{V^j_\xi+[V^j,U^0]\bigg\}
\end{equation}
which is a consequence of \eqref{comp2}.  Compatibility condition (\ref{comp1}) follows directly from the form of $U^0$ and $V^0$ in terms of $g$, i.e. it is not a variational equation of this Lagrangian.  Zakharov and Mikhailov made no reference in \cite{Zakharov1980} to varying the fields $\bar U^i$ and $\bar V^j$ (although, in \cite{Dickey2003,Dickey1990}, Dickey does vary the analog of these fields).  We note that, in the ZM formulation, this would amount to varying a Jordan normal matrix.

\begin{remark}
By letting $\Psi \rightarrow\Phi=g^{-1}\Psi$, letting $U^i \rightarrow \tilde U^i=g^{-1}U^ig$ and letting $V^j \rightarrow \tilde V^j=g^{-1}V^jg$ we can express the auxiliary problem \eqref{eq:Lax} without $U^0$ and $V^0$ terms. This allows us, without loss of generality, to omit $g$ from all ZM related Lagrangians from here on.
\end{remark}

\noindent We shall now change our perspective from the ZM construction, and consider the ZM Lagrangian 

\begin{equation}
\begin{split}
\mathscr{L}_{(\xi\eta)}=\tr\bigg\{\sum_{i=1}^{N_1}(\varphi^i)^{-1}\varphi_\eta^i\bar U^i-\sum_{j=1}^{N_2}(\psi^j)^{-1}\psi_\xi^j\bar V^j-\sum_{i=1}^{N_1}\sum_{j=1}^{N_2}\frac{\psi^j\bar V^j(\psi^j)^{-1}\varphi^i\bar U^i(\varphi^i)^{-1}}{a_i-b_j}\bigg\}
\end{split}
\end{equation}
as our fundamental object.  We impose that $\bar U^i$ and $\bar V^j$ depend only on $\xi$ and $\eta$ respectively.  We no longer impose that $\bar U^i$ and $\bar V^j$ are in Jordan normal form, and now consider them to be fundamental matrix-valued fields.  We now take variational derivatives with respect to all field variables, including $\bar U^i$ and $\bar V^j$.  The variational derivative with respect to $\bar U^i$ reads

\begin{equation}
\frac{\delta \mathscr{L}_{(\xi\eta)}}{\delta \bar U^i}=(\varphi^i)^{-1}\varphi_\eta^i-\sum_{j=1}^{N_2}\frac{(\varphi^i)^{-1}\psi^j\bar V^j(\psi^j)^{-1}\varphi^i}{a_i-b_j}.
\end{equation}
We set this equal to zero and define
\begin{equation}\label{bardef2}
V^j=\psi^j\bar V^j(\psi^j)^{-1}
\end{equation}

and
\begin{equation}\label{eq:UV} 
V=\sum_{j=1}^{N_2}\frac{V^j(\xi,\eta)}{\lambda-b_j}
\end{equation}
to get that

\begin{equation}
\varphi^i_\eta=V|_{\lambda=a_i}\varphi^i
\end{equation}
Similarly, by varying with respect to $\bar V^j$ and setting

\begin{equation}\label{bardef3}
U^i=\varphi^i\bar U^i(\varphi^i)^{-1}
\end{equation}
and
\begin{equation}
U=\sum_{i=1}^{N_1}\frac{U^i(\xi,\eta)}{\lambda-a_i},
\end{equation}
we get that
\begin{equation}
\psi^j_\xi=U|_{\lambda=b_j}\psi^j.
\end{equation}
These relations imply that
\begin{equation}
\begin{split}
U^i_\eta=D_\eta(\varphi^i\bar U^i(\varphi^i)^{-1})&=V|_{\lambda=a_i}\varphi^i\bar U^i(\varphi^i)^{-1}-\varphi^i\bar U^i(\varphi^i)^{-1}V|_{\lambda=a_i}\\
&=[V|_{\lambda=a_i},U^i]
\end{split}
\end{equation}
and similarly
\begin{equation}
V^j_\xi=[U|_{\lambda=b_j},V^j].
\end{equation}
We get precisely the relations \eqref{comp2}.  We have already seen in \eqref{varphivarderv} and \eqref{psivarderv} that the variational derivatives with respect to $\varphi^i$ and $\psi^j$ also give us \eqref{comp2} and the variational derivative with respect to $g$ gives us \eqref{eq:varg} - a corollary of \eqref{comp2}.  Therefore all of these variations give compatible equations.

\begin{remark}
The variational derivative with respect to $\varphi^i$ gives a weaker relation than the variational derivative with respect to $\bar U^i$.  This is due to the non-uniqueness in the choice of $\varphi^i$ when putting $U^i$ into Jordan normal form.  When we re-write these relations in terms of $U^i$, using \eqref{bardef2} this non-uniqueness is removed and we get the same relations in both cases.  A similar statement can be made regarding $\psi^j$ and $\bar V^j$.
\end{remark}

\subsection{Multidimensional Consistency}

One main goal is to incorporate  the ZM Lagrangian into a Lagrangian multiform, each component of which corresponds to two Lax matrices of a Lax multiplet. We will do so for the first nontrivial case of a Lax triplet $(U,V,W)$. In order for this to be possible at all, a necessary property of the triplet is that it produces a multidimensionally consistent  system. Indeed, we will see that a consequence of our construction is that the multiform EL equations form such a consistent system. Therefore, let us introduce a third Lax matrix $W$ and associated independent variable $\nu$ (giving a third part to the auxiliary problem \eqref{eq:Lax} of the form $\Psi_\nu=W\Psi$).  We require that all of the matrices $U$, $V$ and $W$ are functions of three independent variables $\xi, \eta$ and $\nu$.  In addition to the relation

\begin{equation}\label{eq:ZS2}
U_\eta-V_\xi+[U,V]=0.
\end{equation}
that arises when we sum and combine equations \eqref{comp2}, we assume that we have similar relations

\begin{equation}\label{eq:ccheck}
V_\nu-W_\eta+[V,W]=0\quad\text{and}\quad W_\xi-U_\nu+[W,U]=0
\end{equation}
relating $V$ and $W$, and $W$ and $U$.  In order to proceed, we assume that two of the three relations (e.g. the relations \eqref{eq:ccheck}) hold simultaneously and show that the arising compatibility conditions are consistent with the third relation (i.e. the relation \eqref{eq:ZS2}).  If we view the relations \eqref{eq:ccheck} as definitions for the $\eta$ and $\xi$ derivatives of $W$ then we must check that $D_\eta W_\xi-D_\xi W_\eta=0$ when \eqref{eq:ZS2} holds:

\begin{equation}
\begin{split}
D_\eta W_\xi-D_\xi W_\eta&=D_\eta(U_\nu+[U,W])-D_\xi(V_\nu+[V,W])\\
&=U_{\eta\nu}+[U_\eta,W]+[U,W_\eta]-V_{\xi\nu}-[V_\xi,W]-[V,W_\xi]\\
&=U_{\eta\nu}-V_{\xi\nu}+[U_\eta-V_\xi,W]+[U,W_\eta]-[V,W_\xi].
\end{split}
\end{equation}
We use \eqref{eq:ccheck} again to write this as

\begin{equation}
\begin{split}
&U_{\eta\nu}-V_{\xi\nu}+[U_\eta-V_\xi,W]+[U,V_\nu+[V,W]]-[V,U_\nu-[W,U]]\\
=&D_\nu(U_\eta-V_\xi +[U,V])+[U_\eta-V_\xi,W]+ [U,[V,W]]+[V,[W,U]].
\end{split}
\end{equation}
By the Jacobi identity, this is equivalent to
\begin{equation}
D_\nu(U_\eta-V_\xi +[U,V])+[U_\eta-V_\xi+[U,V],W]
\end{equation}
which is zero whenever \eqref{eq:ZS2} holds.

\subsection{A Lagrangian Multiform Structure}\label{ss:LM}

We now introduce the Lagrangian multiform

\begin{equation}\label{mf}
\textsf{L}=\mathscr{L}_{(\xi\eta)}\rm d\xi\wedge d\eta+\mathscr{L}_{(\eta\nu)}\rm d\eta\wedge d\nu+\mathscr{L}_{(\nu\xi)}\rm d\nu\wedge d\xi
\end{equation}
where
\begin{subequations}
\begin{equation}
\begin{split}
\mathscr{L}_{(\xi\eta)}=\tr\bigg\{\sum_{i=1}^{N_1}(\varphi^i)^{-1}\varphi_\eta^i\bar U^i-\sum_{j=1}^{N_2}(\psi^j)^{-1}\psi_\xi^j\bar V^j-\sum_{i=1}^{N_1}\sum_{j=1}^{N_2}\frac{\psi^j\bar V^j(\psi^j)^{-1}\varphi^i\bar U^i(\varphi^i)^{-1}}{a_i-b_j}\bigg\},
\end{split}
\end{equation}

\begin{equation}
\begin{split}
\mathscr{L}_{(\eta\nu)}=\tr\bigg\{\sum_{j=1}^{N_2}(\psi^j)^{-1}\psi_\nu^j\bar V^j-\sum_{k=1}^{N_3}(\chi^k)^{-1}\chi_\eta^k\bar W^k-\sum_{j=1}^{N_2}\sum_{k=1}^{N_3}\frac{\chi^k\bar W^k(\chi^k)^{-1}\psi^j\bar V^j(\psi^i)^{-1}}{b_j-c_k}\bigg\}
\end{split}
\end{equation}
and
\begin{equation}
\begin{split}
\mathscr{L}_{(\nu\xi)}=\tr\bigg\{\sum_{k=1}^{N_3}(\chi^k)^{-1}\chi_\xi^k\bar W^k-\sum_{i=1}^{N_1}(\varphi^i)^{-1}\varphi_\nu^i\bar U^i-\sum_{k=1}^{N_3}\sum_{i=1}^{N_1}\frac{\varphi^i\bar U^i(\varphi^i)^{-1}\chi^k\bar W^k(\chi^k)^{-1}}{c_k-a_i}\bigg\}.
\end{split}
\end{equation}
\end{subequations}
We impose that the $\bar U^i$ only depend on $\xi$, the $\bar V^j$ only depend on $\eta$ and the $\bar W^k$ only depend on $\nu$.  The multiform EL equations of $\sf{L}$ correspond to the criticality of the action

\begin{equation}
S=\int_\sigma \sf{L}
\end{equation}
simultaneously for every surface $\sigma$ in the $\xi,\eta,\nu$ plane and are given by \eqref{eq:ELeqs}.  Since this Lagrangian 2-form depends only on $1^{st}$ order derivatives of the field variables, the multiform EL equations \eqref{eq:ELeqs} reduce to the following:

\begin{description}
\item[$\bullet$ The standard EL equations] 
\end{description}
\begin{subequations}
\begin{equation}\label{eq:clasEL}
\begin{split}
&\frac{\delta\mathscr{L}_{(\xi\eta)}}{\delta\varphi}=0,\quad \frac{\delta \mathscr{L}_{(\xi\eta)}}{\delta\psi}=0,\quad \frac{\delta \mathscr{L}_{(\xi\eta)}}{\delta\chi}=0,\\
 &\frac{\delta \mathscr{L}_{(\xi\eta)}}{\delta\bar U}=0,\quad \frac{\delta \mathscr{L}_{(\xi\eta)}}{\delta\bar V}=0,\quad \frac{\delta \mathscr{L}_{(\xi\eta)}}{\delta\bar W}=0
\end{split}
\end{equation}
and similarly for $\mathscr{L}_{(\eta\nu)}$ and $\mathscr{L}_{(\nu\xi)}$.
\begin{description}
\item[$\bullet$ The first jet one component EL equations] 
\end{description}

\begin{equation}\label{eq:gEL1}
\frac{\delta \mathscr{L}_{(\xi\eta)}}{\delta \varphi_\nu}=0, \quad \frac{\delta \mathscr{L}_{(\xi\eta)}}{\delta \psi_\nu}=0
\end{equation}
and similar relations for cyclic permutations of $\xi, \eta$ and $\nu$.

\begin{description}
\item[$\bullet$ The first jet two component EL equations] 
\end{description}

\begin{equation}\label{eq:gEL2}
\frac{\delta \mathscr{L}_{(\xi\eta)}}{\delta \varphi_\xi}+\frac{\delta \mathscr{L}_{(\eta\nu)}}{\delta \varphi_\nu}=0 ,\quad \frac{\delta \mathscr{L}_{(\eta\nu)}}{\delta \varphi_\eta}+\frac{\delta \mathscr{L}_{(\nu\xi)}}{\delta \varphi_\xi}=0,\quad\frac{\delta \mathscr{L}_{(\nu\xi)}}{\delta \varphi_\nu}+\frac{\delta \mathscr{L}_{(\xi\eta)}}{\delta \varphi_\eta}=0
\end{equation}
\end{subequations}
and similar relations with respect to $\psi$ and $\chi$.

\begin{remark}
Since, in this case, the Lagrangian multiform \textsf{L} has no $2^{nd}$ or higher jet terms, the variational derivatives with respect to any given first jet term are just partial derivatives with respect to that term.
\end{remark}

\begin{theorem}\label{thm}
For the Lagrangian multiform

\begin{equation}\label{ZMMF}
\sf{L}=\mathscr{L}_{(\xi\eta)}  \rm d\xi\wedge d\eta+\mathscr{L}_{(\eta\nu)}d\eta\wedge d\nu+\mathscr{L}_{(\nu\xi)}d\nu\wedge d\xi
\end{equation}

\noindent The relevant EL equations \eqref{eq:clasEL}, \eqref{eq:gEL1} and \eqref{eq:gEL2} yield the multidimensional system of equations given by \eqref{comp2} and the corresponding relations for the matrix $W$.  Furthermore, $\rm d\sf L=0$ on solutions of the multiform Euler-Lagrange equations.

\end{theorem}

\begin{proof} 
We begin by confirming that the Multiform EL equations \eqref{eq:clasEL}, \eqref{eq:gEL1} and \eqref{eq:gEL2} hold.  From varying $\bar U$ and $\bar V$ in $\mathscr{L}_{(\xi\eta)}$ we get
\begin{subequations}
\begin{equation}\label{eq:xietabar}
\varphi^i_\eta=V|_{\lambda =a_i}\varphi^i \quad \text{and}\quad \psi^j_\xi=U|_{\lambda=b_j}\psi^j.
\end{equation}
From varying $\bar V$ and $\bar W$ in $\mathscr{L}_{(\eta\nu)}$ we get

\begin{equation}\label{eq:etanubar}
\psi^j_\nu=W|_{\lambda=b_j}\psi^j\quad \text{and}\quad\chi^k_\eta=V|_{\lambda=c_k}\chi^k.
\end{equation}
From varying $\bar W$ and $\bar U$ in $\mathscr{L}_{(\nu\xi)}$ we get

\begin{equation}\label{eq:nuxibar}
\chi^k_\xi=U|_{\lambda=c_k}\chi^k\quad\text{and}\quad\varphi^i_\nu=W|_{\lambda=a_i}\varphi^i.
\end{equation}
\end{subequations}
From varying $\varphi^i$ and $\psi^j$ in $\mathscr{L}_{(\xi\eta)}$ we get
\begin{subequations}
\begin{equation}
U_\eta^i+\bigg[U^i,\sum_{j=1}^{N_2}\frac{V^j}{a_i-b_j}\bigg]=0\quad \text{and}\quad V_\xi^j+\bigg[V^j,\sum_{i=1}^{N_1}\frac{U^i}{b_j-a_i}\bigg]=0
\end{equation}
which are corollaries of \eqref{eq:xietabar}.  From varying $\psi^j$ and $\chi^k$ in $\mathscr{L}_{(\eta\nu)}$ we get

\begin{equation}
V_\nu^j+\bigg[V^j,\sum_{k=1}^{N_3}\frac{W^k}{b_j-c_k}\bigg]=0\quad \text{and}\quad  W_\eta^k+\bigg[W^k,\sum_{j=1}^{N_2}\frac{V^j}{c_k-b_j}\bigg]=0
\end{equation}
which are corollaries of \eqref{eq:etanubar}.  From varying $\chi^k$ and $\varphi^i$ in $\mathscr{L}_{(\nu \xi)}$ we get

\begin{equation}
W_\xi^k+\bigg[W^k,\sum_{i=1}^{N_1}\frac{U^i}{c_k-a_i}\bigg]=0\quad \text{and}\quad  U_\nu^i+\bigg[U^i,\sum_{k=1}^{N_3}\frac{W^k}{a_i-c_k}\bigg]=0
\end{equation}
\end{subequations}
which are corollaries of \eqref{eq:nuxibar}.  Equations of the type given in \eqref{eq:gEL1} are trivially satisfied since there are no $\nu$ derivatives in $\mathscr{L}_{(\xi\eta)}$, no $\xi$ derivatives in $\mathscr{L}_{(\eta\nu)}$ and no $\eta$ derivatives in $\mathscr{L}_{(\eta\nu)}$.  Equations of the type given in \eqref{eq:gEL2} are also trivially satisfied, in that they do not require the field variables to be critical points of the action in order to hold.\\

\noindent The validity of the relation $\rm d \sf L=0$ for the Lagrangian \eqref{ZMMF} on the solutions of the EL equations is verified by direct computation, the details of which are presented in Appendix \ref{ap:A}.  

\end{proof}

\begin{remark}\label{remlax}
\noindent We notice that the $N_3$ pairs of expressions derived from $\sf{L}$ by varying $\bar W^k$,
\begin{subequations}
\begin{equation}
\chi^k_\xi=U|_{\lambda=c_k}\chi^k\quad \text{and}\quad \chi^k_\eta=V|_{\lambda=c_k}\chi^k
\end{equation}
are precisely the auxiliary problem \eqref{eq:Lax} with $\lambda =c_k$.  Similarly, we can view the $N_1$ expressions involving $\varphi^i$ of the form

\begin{equation}
\varphi^i_\eta=V|_{\lambda =a_i}\varphi^i \quad \text{and}\quad\varphi^i_\nu=W|_{\lambda=a_i}\varphi^i
\end{equation}
that come from varying $\bar U^i$ as an auxiliary problem based on $V$ and $W$ with $\lambda=a_i$ and the $N_2$ expressions involving $\psi^j$ of the form

\begin{equation}
\psi^j_\nu=W|_{\lambda=b_j}\psi^j\quad \text{and}\quad\psi^j_\xi=U|_{\lambda=b_j}\psi^j
\end{equation}
\end{subequations}
that come from varying $\bar V^j$ as an auxiliary problem based on $W$ and $U$ with $\lambda=b_j$.
\end{remark}

\subsection{Lagrangian for the ZM Lax Pair}

\noindent Building on remark \ref{remlax}, in the case of the Lax pair \eqref{eq:Lax} involving $U$ and $V$ with spectral parameter $\lambda$ and associated coordinates $\xi$ and $\eta$, we can view the spectral parameter $\lambda$ as coming from a ``ghost'' direction $\nu$ as one of the poles of the associated Lax matrix $W$.  In this case, the Lagrangian multiform \eqref{mf} can be viewed as the Lagrangian for the Lax pair $U$ and $V$, with the multiform EL equations of the Lagrangian multiform including both the equations of motion of the Lax pair $U$ and $V$ and also the auxiliary problem \eqref{eq:Lax}.  However, it would be just as valid to focus on $V$ and $W$ and consider $\xi$ as the ``ghost'' direction, or to focus on $W$ and $U$ with $\eta$ as the ``ghost'' direction, since the three Lax matrices $U$, $V$ and $W$ along with their respective associated coordinates $\xi$, $\eta$ and $\nu$ all hold equal status within the multiform. Therefore, in the context of this Lagrangian multiform description, rather that considering a Lax pair as consisting of matrices $U$ and $V$ with spectral parameter $\lambda$, it is more satisfactory to consider the Lax triplet $U$, $V$ and $W$.\\

\noindent If we are only interested in the $U$, $V$ auxiliary problem

\begin{equation}\label{eq:ap2}
\Psi_\xi=U(\xi,\eta,\lambda)\Psi, \quad \Psi_\eta=V(\xi,\eta,\lambda)\Psi,
\end{equation}
and want to cast this in the multiform structure of Section \ref{ss:LM} then it is necessary to introduce a ``ghost'' variable $\nu$ and require that all field variables now have a $\nu$ dependence.  We must also introduce the additional Lax matrix $W$ relating to the ``ghost'' direction $\nu$.  These are required in the Lagrangian in order to have a closed 2-form.  The multiform EL equations from such a Lagrangian 2-form will have a $\nu$ dependence.  We will go on to show that any set of $\nu$ dependent solutions can be reduced to a set of $\nu$ independent solutions, thereby obtaining precisely the auxiliary problem \eqref{eq:ap2} and the associated compatibility conditions depending only on $\xi$ and $\eta$ from our Lagrangian 2-form.\\

\noindent We take our Lagrangian $\sf{L}$$[\varphi,\psi,\chi, \bar U, \bar V, \bar W;\lambda]$ to be

\begin{equation}\label{eq:lmlax}
\begin{split}
\textsf{L}=&\big(\sum_{\i=1}^{N_1}(\varphi^i)^{-1}\varphi^i_\eta \bar U^i-\sum_{j=1}^{N_2}(\psi^j)^{-1}\psi^j_\xi\bar V^j-\sum_{\i=1}^{N_1}\sum_{j=1}^{N_2}\frac{\psi^j \bar V^j (\psi^j)^{-1}\varphi^i \bar U^i(\varphi^i)^{-1}}{a_i-b_j}\big)d\xi \wedge d\eta\\
&+\big(\sum_{j=1}^{N_2}(\psi^j)^{-1}\psi^j_\nu\bar V^j-\chi^{-1}\chi_\eta\bar W-\sum_{j=1}^{N_2}\frac{\chi\bar W\chi^{-1}\psi^j\bar V^j(\psi^j)^{-1}}{b_j-\lambda}\big)d\eta \wedge d\nu\\
&+\big(\chi^{-1}\chi_\xi\bar W-\sum_{\i=1}^{N_1}(\varphi^i)^{-1}\varphi^i_\nu \bar U^i-\sum_{\i=1}^{N_1}\frac{\varphi^i\bar U^i(\varphi^i)^{-1}\chi\bar W\chi^{-1}}{\lambda-a_i}\big)d\nu\wedge d\xi.
\end{split}
\end{equation}

\noindent This Lagrangian 2-form is special case of the multiform \eqref{mf} where the matrix $W$ has a single pole at $\lambda$.  In accordance with Theorem \ref{thm} the multiform equations of motion given by this multiform are
\begin{subequations}
\begin{equation}\label{eq:chi1}
\chi_\xi=U\chi\quad \text{and}\quad \chi_\eta=V\chi
\end{equation}

\begin{equation}\label{eq:phi1}
\varphi^i_\eta=V|_{\lambda =a_i}\varphi^i \quad \text{and}\quad\varphi^i_\nu=W|_{\lambda=a_i}\varphi^i
\end{equation}

\begin{equation}\label{eq:psi1}
\psi^j_\nu=W|_{\lambda=b_j}\psi^j\quad \text{and}\quad\psi^j_\xi=U|_{\lambda=b_j}\psi^j
\end{equation}
\end{subequations}
and corollaries thereof, including
\begin{subequations}
\begin{equation}
U_\eta^i+\bigg[U^i,\sum_{j=1}^{N_2}\frac{V^j}{a_i-b_j}\bigg]=0\quad \text{and}\quad V_\xi^j+\bigg[V^j,\sum_{i=1}^{N_1}\frac{U^i}{b_j-a_i}\bigg]=0
\end{equation}
\begin{equation}
V_\nu^j+\bigg[V^j,\frac{W^1}{b_j-\lambda}\bigg]=0\quad \text{and}\quad  W_\eta^1+\bigg[W^1,\sum_{j=1}^{N_2}\frac{V^j}{\lambda-b_j}\bigg]=0
\end{equation}
\begin{equation}
W_\xi^1+\bigg[W^1,\sum_{i=1}^{N_1}\frac{U^i}{\lambda-a_i}\bigg]=0\quad \text{and}\quad  U_\nu^i+\bigg[U^i,\frac{W^1}{a_i-\lambda}\bigg]=0.
\end{equation}
\end{subequations}
At this stage, our equations of motion contain $\nu$ which does not feature in the $U$, $V$ Lax pair. However, if the matrices $\bar U$, $\bar V$, $\bar W$, $\varphi^i$, $\psi^j$ and $\chi$ satisfy these equations, then there is also a solution with the same $\bar U$ and $\bar V$ but with $\bar W=0$.  In this case the second equation of \eqref{eq:phi1} and the first equation of \eqref{eq:psi1} tell us that $\varphi^i$ and $\psi^j$ no longer depend on $\nu$, i.e. we can think of these as the $\varphi^i$ and $\psi^j$ of the original solution, with $\nu=\nu_0$, a constant.  The first equation of \eqref{eq:phi1} and the second equation of \eqref{eq:psi1} are simply the definitions of $\varphi^i$ and $\psi^j$ which hold for $\nu=\nu_0$.  Then \eqref{eq:chi1} is precisely the auxiliary problem for $U$ and $V$, which no longer depends upon $\nu$.  Thus, the only remaining relations that are non-zero are

\begin{equation}
\chi_\xi=U\chi\quad \text{and}\quad \chi_\eta=V\chi
\end{equation}
the auxiliary problem based on $U$ and $V$,

\begin{equation}
\varphi^i_\eta=V|_{\lambda =a_i}\varphi^i \quad \text{and}\quad\psi^j_\xi=U|_{\lambda=b_j}\psi^j
\end{equation}
the defining relations for $\varphi^i$ and $\psi^j$
and

\begin{equation}
U_\eta^i+[U^i,\sum_{j=1}^{N_2}\frac{V^j}{a_i-b_j}]=0\quad \text{and}\quad V_\xi^j+[V^j,\sum_{i=1}^{N_1}\frac{U^i}{b_j-a_i}]=0
\end{equation}
the equations of motion for $U^i$ and $V^j$.  All of these relations now only depend upon $\xi$ and $\eta$.  Therefore, the Lagrangian multiform \eqref{eq:lmlax} can be considered the Lagrangian for the Lax pair $U$ and $V$.  We can summarise this result in the following theorem.

\begin{theorem}

The Lagrangian 2-form $\sf{L}$$(\varphi,\psi,\chi, \bar U, \bar V, \bar W,g;\lambda)$ given by \eqref{eq:lmlax} is a Lagrangian for the Lax pair $U$ and $V$.  When we take the multiform EL equations and set $\bar W=0$ our equations of motion are the auxiliary problem

\begin{equation}
\chi_\xi=U\chi\quad \text{and}\quad \chi_\eta=V\chi
\end{equation}
for $U$ and $V$ and the equations of motion

\begin{equation}
U_\eta^i+[U^i,\sum_{j=1}^{N_2}\frac{V^j}{a_i-b_j}]=0\quad \text{and}\quad V_\xi^j+[V^j,\sum_{i=1}^{N_1}\frac{U^i}{b_j-a_i}]=0
\end{equation}
 corresponding to the compatibility conditions of this auxiliary problem.

\end{theorem}

\section{Matrix AKNS Hierarchy}

As a specific example of the general construction, we present here the case of the single-pole Lax pair which, with appropriate choice of variables, can be viewed as a generating model for the generalized, i.e. $N\times N$ matrix generalization, of the famous AKNS hierarchy of \cite{AKNS1974}.\\ 

\subsection{An Integrable $N \times N$ Hierarchy and its ZM Lagrangian}

\noindent We begin by introducing co-ordinates $x_i$ for $i=1,\ldots,\infty$ and we define the derivatives with respect to $\xi$  and $\eta$ such that

\begin{equation}
\partial_\xi=\sum_{i=0}^{\infty}\frac{1}{a^{i+1}}\frac{\partial}{\partial x_i}\quad\text{and}\quad\partial_\eta=\sum_{j=0}^{\infty}\frac{1}{b^{j+1}}\frac{\partial}{\partial x_j}
\end{equation}

\noindent and apply this to form a Lax pair and auxiliary problem with a single simple pole

\begin{equation}\label{eq:pair1}
\Psi_\xi=\frac{U(\xi,\eta)}{\lambda - a}\Psi, \quad \Psi_\eta=\frac{V(\xi,\eta)}{\lambda - b}\Psi,
\end{equation}
i.e. the ZM auxiliary problem with $N_1=1$ and $N_2=1$.  This gives rise to the compatibility conditions

\begin{equation}\label{em}
U_\eta=V_\xi=\frac{[V,U]}{a-b}.
\end{equation}
By the ZM method outlined in Section \ref{s:MLE}, this has the Lagrangian
\begin{subequations}
\begin{equation}\label{eq:WZWLag}
\mathscr{L}_{(\xi\eta)}=\tr\bigg\{\varphi^{-1}\varphi_\eta\bar U-\psi^{-1}\psi_\xi\bar V-\frac{\psi\bar V\psi^{-1}\varphi\bar U\varphi^{-1}}{a-b}\bigg\}
\end{equation}

\noindent We can now introduce the co-ordinate $\nu$, the associated matrix $\bar W(\nu)$ and parameter $c$ to form two further Lagrangians

\begin{equation}
\mathscr{L}_{(\eta\nu)}=\tr\bigg\{\psi^{-1}\psi_\nu\bar V-\chi^{-1}\chi_\eta\bar W-\frac{\chi\bar W\chi^{-1}\psi\bar V\psi^{-1}}{b-c}\bigg\}
\end{equation}
and
\begin{equation}
\mathscr{L}_{(\nu\xi)}=\tr\bigg\{\chi^{-1}\chi_\xi\bar W-\varphi^{-1}\varphi_\nu\bar U-\frac{\varphi\bar U\varphi^{-1}\chi\bar W\chi^{-1}}{c-a}\bigg\}
\end{equation}
to form the Lagrangian multiform
\begin{equation}
\mathscr{L}_{(\xi\eta)}d\xi\wedge d\eta+\mathscr{L}_{(\eta\nu)}d\eta\wedge d\nu+\mathscr{L}_{(\nu\xi)}d\nu\wedge d\xi.
\end{equation}
\end{subequations}
By Theorem \ref{thm}, this Lagrangian multiform is closed on solutions of this system and has Multiform EL equations that include \eqref{em} when we let $U=\varphi \bar U \varphi^{-1}$ and $V=\psi \bar V\psi^{-1}$,  i.e. we have a Lagrangian multiform structure for this system.\\

\noindent Since, on the equations of motion, $U_\eta=V_\xi$, there exists a matrix $H$ such that $U=H_\xi$ and $V=H_\eta$.  Expressing (\ref{em}) in terms of $H$,  we get

\begin{equation}\label{eq:Hxi}
H_{\xi\eta}=\frac{[H_\eta,H_\xi]}{a-b}.
\end{equation}
A conventional Lagrangian that gives this expression directly is given in \cite{Nijhoff1987}.  When we expand the $\xi$ and $\eta$ derivatives in terms of the $x_i$ co-ordinates this gives us

\begin{equation}\label{eq:Hrel}
H_{x_{i}x_{j-1}}-H_{x_{i-1}x_{j}}=[H_{x_{j-1}},H_{x_{i-1}}],
\end{equation}
an integrable $N\times N$ matrix system \cite{Nijhoff1987,Nijhoff1988}. We will show that, in the $2 \times 2$ case, this contains the AKNS hierarchy; this particular case, and the underlying Kac-Moody algebra structure were treated in \cite{FLASCHKA1983}, where in particular the corresponding symplectic forms were given.\\

\noindent We define the matrix

\begin{equation}
Q_i:=-\partial_{x_{i-1}}H\quad \text{for }i\geq 1
\end{equation}
so \eqref{eq:Hrel} becomes

\begin{equation}\label{eq:Qeq1}
\partial_{x_j}Q_{i}-\partial_{x_i}Q_{j}=[Q_j,Q_i]
\end{equation}
and since partial derivatives of $H$ with respect to the $x_i$ co-ordinates commute, we also have that

\begin{equation}\label{eq:Qeq2}
\partial_{x_{i}}Q_j=\partial_{x_{j-1}}Q_{i+1}.
\end{equation}
If we define $Q_0$ to be a constant matrix then \eqref{eq:Qeq1} and \eqref{eq:Qeq2} give us the additional relation

\begin{equation}\label{eq:Qeq3}
[Q_0,Q_{k+1}]+[Q_{1},Q_k]=\partial_{x_1}Q_k.
\end{equation}
The relations \eqref{eq:Qeq1},\eqref{eq:Qeq2} and \eqref{eq:Qeq3}are used recursively to find $Q_i$ for all $i$.  In the case of the AKNS hierarchy, we take the $Q_i$ to be $2\times2$ matrices and define

\begin{equation}
Q_0=
\begin{pmatrix}
-i & 0\\
0 & i
\end{pmatrix},\quad 
Q_1=
\begin{pmatrix}
0 & q\\
r & 0
\end{pmatrix}
\end{equation}
where $q$ and $r$ are functions of the $x_i$ co-ordinates.  We are now able to follow the procedure outlined in \cite{FLASCHKA1983} and use \eqref{eq:Qeq1},\eqref{eq:Qeq2} and \eqref{eq:Qeq3} recursively to find the 
$Q_i$,  e.g.

\begin{equation}
Q_2=\frac{i}{2}
\begin{pmatrix}
-qr & q_{x_1}\\
-r_{x_1} & qr
\end{pmatrix}, \quad
Q_3=-\frac{1}{4}
\begin{pmatrix}
qr_{x_1}-rq_{x_1}& q_{x_1x_1}-2q^2r\\
r_{x_1x_1}-2qr^2&-qr_{x_1}+rq_{x_1}
\end{pmatrix},\quad \ldots
\end{equation}
The equations of the AKNS hierarchy are given by

\begin{equation}
\partial_{x_N}Q_{1}-\partial_{x_1}Q_{N}=[Q_{N},Q_{1}]
\end{equation}
i.e. equation \eqref{eq:Qeq1} with $i=1$.\\

\subsection{A Scalar AKNS Multiform}

Scalar Lagrangians for the AKNS hierarchy also possess a Lagrangian multiform structure.  The $\mathscr{L}_{(x_1x_2)}$ and $\mathscr{L}_{(x_3x_1)}$ AKNS Lagrangians, see e.g. \cite{AVAN2016415} are as follows:

\begin{equation}
\mathscr{L}_{(x_1x_2)}= \frac{1}{2}(rq_{x_2}-qr_{x_2})+\frac{i}{2}q_{x_1}r_{x_1}+\frac{i}{2}q^2r^2
\end{equation}
giving equations of motion
\begin{subequations}
\begin{equation}
q_{x_2}=\frac{i}{2}q_{x_1x_1}-iq^2r
\end{equation}
and
\begin{equation}
r_{x_2}=-\frac{i}{2}r_{x_1x_1}+ir^2q
\end{equation}
\end{subequations}
corresponding to $\dfrac{\delta \mathscr{L}_{(x_1x_2)}}{\delta r}=0$ and $\dfrac{\delta \mathscr{L}_{(x_1x_2)}}{\delta q}=0$ respectively.  These are identical to the equations given by the off diagonal entries of

\begin{equation}
\partial_{x_2}Q_{1}-\partial_{x_1}Q_{2}=[Q_{2},Q_{1}].
\end{equation}

\begin{equation}
\mathscr{L}_{(x_3x_1)}= \frac{1}{2}(qr_{x_3}-rq_{x_3})+\frac{1}{8}(r_{x_1}q_{x_1x_1}-q_{x_1}r_{x_1x_1})+\frac{3}{8}qr(rq_{x_1}-qr_{x_1}).
\end{equation}
giving equations of motion
\begin{subequations}
\begin{equation}
q_{x_3}=\frac{3}{2}qrq_{x_1}-\frac{1}{4}q_{x_1x_1x_1}
\end{equation}
and
\begin{equation}
r_{x_3}=\frac{3}{2}rqr_{x_1}-\frac{1}{4}r_{x_1x_1x_1}
\end{equation}
\end{subequations}
corresponding to $\dfrac{\delta \mathscr{L}_{(x_3x_1)}}{\delta r}=0$ and $\dfrac{\delta \mathscr{L}_{(x_3x_1)}}{\delta q}=0$ respectively.  These are identical to the equations given by the off diagonal entries of

\begin{equation}
\partial_{x_3}Q_{1}-\partial_{x_1}Q_{3}=[Q_{3},Q_{1}].
\end{equation}
From the requirement that $\delta \sf{dL}=0$ for the Lagrangian 2-form 
\begin{equation}
\textsf{L}=\mathscr{L}_{(x_1x_2)}\textrm{ d} x_1\wedge \textrm{ d}x_2+\mathscr{L}_{(x_2x_3)}\textrm{ d}x_2\wedge \textrm{ d}x_3+\mathscr{L}_{(x_3x_1)}\textrm{ d}x_3\wedge \textrm{ d}x_1,
\end{equation}
we are able to derive the $\mathscr{L}_{(x_2x_3)}$ Lagrangian as

\begin{equation}
\begin{split}
\mathscr{L}_{(x_2x_3)}=&\frac{1}{4}(q_{x_2}r_{x_1x_1}-r_{x_2}q_{x_1x_1})-\frac{i}{2}(q_{x_3}r_{x_1}+r_{x_3}q_{x_1})+\frac{1}{8}(q_{x_1}r_{x_1x_2}-r_{x_1}q_{x_1x_2})+\frac{3}{8}qr(qr_{x_2}-rq_{x_2})\\
&-\frac{i}{8}q_{x_1x_1}r_{x_1x_1}+\frac{i}{4}qr(qr_{x_1x_1}+rq_{x_1x_1})-\frac{i}{8}(q^2r_{x_1}^2+r^2q_{x_1}^2)+\frac{i}{4}qrq_{x_1}r_{x_1}-\frac{i}{2}q^3r^3.
\end{split}
\end{equation}
The equations of motion for this $\mathscr{L}_{(x_2x_3)}$ Lagrangian are 
\begin{subequations}
\begin{equation}
\frac{i}{2}r_{x_1x_3}+\frac{3}{4}qrr_{x_2}-\frac{1}{4}r_{x_1x_1x_2}+\frac{3}{4}qrr_{x_2}+\frac{i}{2}qrr_{x_1x_1}+\frac{i}{4}r^2q_{x_1x_1}-\frac{i}{4}qr_{x_1}^2+\frac{i}{4}rq_{x_1}r_{x_1}-\frac{3i}{2}q^2r^3=0
\end{equation}
and
\begin{equation}
-\frac{i}{2}q_{x_1x_3}+\frac{3}{4}rqq_{x_2}-\frac{1}{4}q_{x_1x_1x_2}+\frac{3}{4}rqq_{x_2}-\frac{i}{2}rqq_{x_1x_1}-\frac{i}{4}q^2r_{x_1x_1}+\frac{i}{4}rq_{x_1}^2-\frac{i}{4}qr_{x_1}r_{x_1}+\frac{3i}{2}r^2q^3=0
\end{equation}
\end{subequations}
corresponding to $\dfrac{\delta \mathscr{L}_{(x_2x_3)}}{\delta q}=0$ and $\dfrac{\delta \mathscr{L}_{(x_2x_3)}}{\delta r}=0$ respectively.  These do not correspond directly to the equations given by

\begin{equation}
\partial_{x_3}Q_{2}-\partial_{x_2}Q_{3}=[Q_{3},Q_{2}].
\end{equation}
but are equivalent to them modulo the equations of motion of $\mathscr{L}_{(x_1x_2)}$ and $\mathscr{L}_{(x_3x_1)}$.  By construction, this Lagrangian multiform satisfies $\delta \rm d\sf{L}=0$ and $\rm d \sf L=0$ for $q$ and $r$ satisfying the equations of the AKNS hierarchy. Consequently, it obeys the multiform EL equations \eqref{eq:ELeqs}.  This is the first time that scalar Lagrangians of the AKNS hierarchy have been shown to fit into a Lagrangian multiform description.  This result is analogous to the Lagrangian multiform relating to the KdV and sine-Gordon equations, presented in \cite{Suris2016}.\\

\section{Conclusion}

\noindent Using  the  method  outlined  in  this  paper, one is able to construct a Lagrangian multiform  structure for systems with Lax  pairs  in  the  appropriate  form, and in so doing, find a Lagrangian for the Lax pair itself. Lagrangians in the case of Lax pairs with higher-order poles were given by Dickey in \cite{Dickey2003}, and it is to be expected that those can be extended to a Lagrangian multiform structure. The {\em generating PDEs} introduced in \cite{NIJHOFF2000147,Tongas2005} which are associated with non-isospectral Lax pairs, possess Lagrangians of the required form, cf. also \cite{Lobb2009}. Furthermore, we expect that the universal symplectic form of Krichever and Phong, \cite{Krichever1997,Krichever1998} associated with Lax operators could play a role in the construction of Lagrangians possessing a multiform structure. 

\appendix

\section{Proof of Theorem \ref{thm:ELeq}}\label{proof1}

We begin by introducing the notation

\begin{equation}
\boldsymbol{\varphi}_{(a,b,c)}:=\bigg(\frac{\partial}{\partial t_1}\bigg)^a\bigg(\frac{\partial}{\partial t_2}\bigg)^b\bigg(\frac{\partial}{\partial t_3}\bigg)^c\boldsymbol{\varphi} \quad\text{for } a,b,c \in \mathbb{Z}^+
\end {equation}

\noindent and consider the Lagrangian 2-form

\begin{equation}
\sf{L}=\mathscr{L}_{(12)}dt_1\wedge dt_2+\mathscr{L}_{(23)}dt_2\wedge dt_3+\mathscr{L}_{(31)}dt_3\wedge dt_1
\end{equation}

\noindent Let $B$ be an arbitrary three dimensional ball with surface $\partial B$.  We consider the action functional $S$ over the closed surface $\partial B$ such that

\begin{equation}
S[\boldsymbol{\varphi}]=\oint_{\partial B}\sf{L}
\end{equation}
We then apply Stokes' theorem to write $S$ in terms of an integral over B:

\begin{equation}
S[\boldsymbol{\varphi}]=\int_ B\sf{dL}
\end{equation}
and we look for solutions of

\begin{equation}\label{deltaS1}
\delta S=\int_{B}\delta\sf{d}\sf{L}=0
\end{equation}
Since this must hold for arbitrary variations (i.e. with no boundary constraints) for every arbitrary ball $B$, it follows that on solutions $\boldsymbol{\varphi}$ of our system, $\delta \sf{dL}=0$.  Up to this point, we have used the same argument as the one given in the proof of Proposition 2.2 in \cite{Suris2016}.  We now define $\nabla^n$ to be the vector of all $n^{th}$ order derivatives with respect to $t_1, t_2$ and $t_3$ (e.g. $\nabla^2= (\partial_{t_1t_1},\partial_{t_1t_2},\partial_{t_1t_3},\partial_{t_2t_2},\partial_{t_2t_3},\partial_{t_3t_3})^t$).  If our Lagrangian 2-form has terms up to $N^{th}$ order derivatives of our field variable $\boldsymbol{\varphi}$ then we can expand $\delta \sf{dL}$:

\begin{equation}
\delta \sf{dL}=\sum_{i=0}^{N+1}\partial_{\nabla^i \boldsymbol{\varphi}}\sf{dL}\cdot\delta \nabla^i \boldsymbol{\varphi}=0.
\end{equation}
Our sum is up to $N+1$ since $\sf{dL}$ will contain $N+1^{th}$ order derivatives of $\boldsymbol{\varphi}$.  Since $\delta \sf{dL}=0$, each coefficient of $\delta \boldsymbol{\varphi}_{(a,b,c)}$ is zero on solutions of our system.  As a result, it is also true that

\begin{equation}
\begin{split}
\delta \sf{dL}&=\partial_{\nabla^{N+1} \boldsymbol{\varphi}}\sf{dL}\cdot\delta \nabla^{N+1} \boldsymbol{\varphi}\\
&=\partial_{\nabla^N \boldsymbol{\varphi}}\mathcal{L}_{(23)}\cdot\delta \nabla^N \boldsymbol{\varphi}_{t_1}+\partial_{\nabla^N \boldsymbol{\varphi}}\mathcal{L}_{(31)}\cdot\delta \nabla^N \boldsymbol{\varphi}_{t_2}+\partial_{\nabla^N \boldsymbol{\varphi}}\mathcal{L}_{(12)}\cdot\delta \nabla^N \boldsymbol{\varphi}_{t_3}
\end{split}
\end{equation}
in terms of the constituent $\mathscr{L}_{(i,j)}$.  As a result,

\begin{equation}\label{ndelta}
\delta S=\int_{B}(\partial_{\nabla^N \boldsymbol{\varphi}}\mathcal{L}_{(23)}\cdot\delta \nabla^N \boldsymbol{\varphi}_{t_1}+\partial_{\nabla^N \boldsymbol{\varphi}}\mathcal{L}_{(31)}\cdot\delta \nabla^N \boldsymbol{\varphi}_{t_2}+\partial_{\nabla^N \boldsymbol{\varphi}}\mathcal{L}_{(12)}\cdot\delta \nabla^N \boldsymbol{\varphi}_{t_3})dt_1\wedge dt_2\wedge dt_3
\end{equation}
where the coefficient of each $N+1^{th}$ order derivative of $\delta \boldsymbol{\varphi}$ is zero.  I.e. we get $\dfrac{(N+2)(N+3)}{2}$ relations on the constituent $\mathscr{L}_{(ij)}$, one for each of the  $N+1^{th}$ order derivatives of $\delta \boldsymbol{\varphi}$.  We can then get further relations on the constituent $\mathscr{L}_{(ij)}$ by performing integration by parts on \eqref{ndelta} in order to find the coefficients of lower order derivatives of $\delta \boldsymbol{\varphi}$, which must also be identically zero.  We will use an inductive argument to show the general form of these relations.  In our current notation,  the variational derivative takes the form

\begin{equation}
\frac{\delta\mathscr{L}_{(12)}}{\delta \boldsymbol{\varphi}_{(a,b,c)}}=\sum_{\alpha ,\beta \geq 0}(-1)^{\alpha + \beta} D_1^\alpha D_2^\beta\frac{\partial \mathscr{L}_{(12)}}{\partial \boldsymbol{\varphi}_{(a+\alpha,b+\beta,c)}} \quad \text{for } a,b,c\in \mathbb{Z}^+
\end{equation}
and similarly for $\dfrac{\delta\mathscr{L}_{(23)}}{\delta \boldsymbol{\varphi}_{(a,b,c)}}$ and $\dfrac{\delta\mathscr{L}_{(31)}}{\delta \boldsymbol{\varphi}_{(a,b,c)}}$.  In the case where one or more of $a,b,c$ is negative, $\dfrac{\delta\mathscr{L}_{(ij)}}{\delta \boldsymbol{\varphi}_{(a,b,c)}}:=0$.  We will refer to the sum $a+b+c$ as the \textbf{order} of the variational derivative.\\

\noindent The scheme of this proof from here is to first use the form of \eqref{ndelta} to show that 
\begin{equation}\label{trip}
\frac{\delta\mathscr{L}_{(12)}}{\delta \boldsymbol{\varphi}_{(l,m,n-1)}}+\frac{\delta\mathscr{L}_{(23)}}{\delta \boldsymbol{\varphi}_{(l-1,m,n)}}+\frac{\delta\mathscr{L}_{(31)}}{\delta \boldsymbol{\varphi}_{(l,m-1,n)}}=0
\end{equation}
holds for $l+m+n\geq N$, and then use an inductive argument to show that if this holds for $l+m+n\geq M$ then it also holds for $l+m+n=M-1$.\\

\noindent We begin by noticing that by setting the coefficients of each $N+1^{th}$ jet derivative of $\delta \boldsymbol{\varphi}$ in \eqref{ndelta} equal to zero, we get the relations \eqref{trip} for $l+m+n=N$.  We also notice that for $l+m+n>N$ the relations \eqref{trip} hold since all terms are zero.  We now assume that for $l+m+n>M$ the coefficient of $\delta \boldsymbol{\varphi}_{(l,m,n)}$ in \eqref{ndelta} is

\begin{equation}\label{eq:ELeqa}
\frac{\delta\mathscr{L}_{(12)}}{\delta \boldsymbol{\varphi}_{(l,m,n-1)}}+\frac{\delta\mathscr{L}_{(23)}}{\delta \boldsymbol{\varphi}_{(l-1,m,n)}}+\frac{\delta\mathscr{L}_{(31)}}{\delta \boldsymbol{\varphi}_{(l,m-1,n)}}
\end{equation}
and that the coefficient of $\delta \boldsymbol{\varphi}_{(l,m,n)}$ for $l+m+n\leq M$ is zero.  We also assume that \eqref{trip} holds for order greater than or equal to $M$.  We now proceed by finding the coefficient of $\delta \boldsymbol{\varphi}_{(l,m,n)}$ in the case where $l+m+n=M$ and deriving the relations corresponding to setting this coefficient equal to zero.  First, we note that 

\begin{equation}\label{split}
\frac{\delta\mathscr{L}_{(12)}}{\delta \boldsymbol{\varphi}_{(a,b,c)}}=\frac{\partial \mathscr{L}_{(12)}}{\partial \boldsymbol{\varphi}_{(a,b,c)}}-D_1\frac{\delta\mathscr{L}_{(12)}}{\delta \boldsymbol{\varphi}_{(a+1,b,c)}}-D_2\frac{\delta\mathscr{L}_{(12)}}{\delta \boldsymbol{\varphi}_{(a,b+1,c)}}-D_1D_2\frac{\delta\mathscr{L}_{(12)}}{\delta \boldsymbol{\varphi}_{(a+1,b+1,c)}}
\end{equation}
Along with similar relations for $\mathscr{L}_{(23)}$ and $\mathscr{L}_{(31)}$.  By our inductive hypothesis we have that

\begin{equation}
\begin{split}
&\text{the coefficient of }\delta \boldsymbol{\varphi}_{(l+1,m,n)}\text{ is }\frac{\delta\mathscr{L}_{(12)}}{\delta \boldsymbol{\varphi}_{(l+1,m,n-1)}}+\frac{\delta\mathscr{L}_{(23)}}{\delta \boldsymbol{\varphi}_{(l,m,n)}}+\frac{\delta\mathscr{L}_{(31)}}{\delta \boldsymbol{\varphi}_{(l+1,m-1,n)}}\\
&\text{the coefficient of }\delta \boldsymbol{\varphi}_{(l,m+1,n)}\text{ is }\frac{\delta\mathscr{L}_{(12)}}{\delta \boldsymbol{\varphi}_{(l,m+1,n-1)}}+\frac{\delta\mathscr{L}_{(23)}}{\delta \boldsymbol{\varphi}_{(l-1,m+1,n)}}+\frac{\delta\mathscr{L}_{(31)}}{\delta \boldsymbol{\varphi}_{(l,m,n)}}\\
&\text{the coefficient of }\delta \boldsymbol{\varphi}_{(l,m,n+1)}\text{ is }\frac{\delta\mathscr{L}_{(12)}}{\delta \boldsymbol{\varphi}_{(l,m,n)}}+\frac{\delta\mathscr{L}_{(23)}}{\delta \boldsymbol{\varphi}_{(l-1,m,n+1)}}+\frac{\delta\mathscr{L}_{(31)}}{\delta \boldsymbol{\varphi}_{(l,m-1,n+1)}}.
\end{split}
\end{equation}
We will now use integration by parts to find the coefficient of $\delta \boldsymbol{\varphi}_{(l,m,n)}$.  We notice that, since all of the coefficients above are identically zero, there will be no contribution from boundary terms.  Therefore, the coefficient of $\delta \boldsymbol{\varphi}_{(l,m,n)}$ in \eqref{ndelta} is

\begin{equation}
\begin{split}
&-D_1\bigg(\frac{\delta\mathscr{L}_{(12)}}{\delta \boldsymbol{\varphi}_{(l+1,m,n-1)}}+\frac{\delta\mathscr{L}_{(23)}}{\delta \boldsymbol{\varphi}_{(l,m,n)}}+\frac{\delta\mathscr{L}_{(31)}}{\delta \boldsymbol{\varphi}_{(l+1,m-1,n)}}\bigg)\\
&-D_2\bigg(\frac{\delta\mathscr{L}_{(12)}}{\delta \boldsymbol{\varphi}_{(l,m+1,n-1)}}+\frac{\delta\mathscr{L}_{(23)}}{\delta \boldsymbol{\varphi}_{(l-1,m+1,n)}}+\frac{\delta\mathscr{L}_{(31)}}{\delta \boldsymbol{\varphi}_{(l,m,n)}}\bigg)\\
&-D_3\bigg(\frac{\delta\mathscr{L}_{(12)}}{\delta \boldsymbol{\varphi}_{(l,m,n)}}+\frac{\delta\mathscr{L}_{(23)}}{\delta \boldsymbol{\varphi}_{(l-1,m,n+1)}}+\frac{\delta\mathscr{L}_{(31)}}{\delta \boldsymbol{\varphi}_{(l,m-1,n+1)}}\bigg)
\end{split}
\end{equation}
which,by \eqref{split} is equal to

\begin{equation}\label{expansion}
\begin{split}
&-D_1\frac{\delta\mathscr{L}_{(12)}}{\delta \boldsymbol{\varphi}_{(l+1,m,n-1)}}-\frac{\partial}{\partial \boldsymbol{\varphi}_{(l,m,n)}}D_1\mathscr{L}_{(23)}+\frac{\partial \mathscr{L}_{(23)}}{\partial \boldsymbol{\varphi}_{(l-1,m,n)}} +D_1D_2\frac{\delta\mathscr{L}_{(23)}}{\delta \boldsymbol{\varphi}_{(l,m+1,n)}}+D_1D_3\frac{\delta\mathscr{L}_{(23)}}{\delta \boldsymbol{\varphi}_{(l,m,n+1)}}\\&+D_1D_2D_3\frac{\delta\mathscr{L}_{(23)}}{\delta \boldsymbol{\varphi}_{(l,m+1,n+1)}}-D_1\frac{\delta\mathscr{L}_{(31)}}{\delta \boldsymbol{\varphi}_{(l+1,m-1,n)}}\\
&-D_2\frac{\delta\mathscr{L}_{(12)}}{\delta \boldsymbol{\varphi}_{(l,m+1,n-1)}}-D_2\frac{\delta\mathscr{L}_{(23)}}{\delta \boldsymbol{\varphi}_{(l-1,m+1,n)}}-\frac{\partial}{\partial \boldsymbol{\varphi}_{(l,m,n)}}D_2\mathscr{L}_{(31)}+\frac{\partial \mathscr{L}_{(31)}}{\partial \boldsymbol{\varphi}_{(l,m-1,n)}} +D_1D_2\frac{\delta\mathscr{L}_{(31)}}{\delta \boldsymbol{\varphi}_{(l+1,m,n)}}\\&+D_2D_3\frac{\delta\mathscr{L}_{(31)}}{\delta \boldsymbol{\varphi}_{(l,m,n+1)}}+D_1D_2D_3\frac{\delta\mathscr{L}_{(23)}}{\delta \boldsymbol{\varphi}_{(l+1,m,n+1)}}\\
&-\frac{\partial}{\partial \boldsymbol{\varphi}_{(l,m,n)}}D_3\mathscr{L}_{(12)}+\frac{\partial \mathscr{L}_{(12)}}{\partial \boldsymbol{\varphi}_{(l,m,n-1)}} +D_1D_3\frac{\delta\mathscr{L}_{(12)}}{\delta \boldsymbol{\varphi}_{(l+1,m,n)}}+D_2D_3\frac{\delta\mathscr{L}_{(12)}}{\delta \boldsymbol{\varphi}_{(l,m+1,n)}}+D_1D_2D_3\frac{\delta\mathscr{L}_{(12)}}{\delta \boldsymbol{\varphi}_{(l+1,m+1,n)}}\\
&-D_3\frac{\delta\mathscr{L}_{(23)}}{\delta \boldsymbol{\varphi}_{(l-1,m,n+1)}}-D_3\frac{\delta\mathscr{L}_{(31)}}{\delta \boldsymbol{\varphi}_{(l,m-1,n+1)}}.
\end{split}
\end{equation}
Here we have used that
\begin{equation}\label{comm}
\begin{split}
D_1\partial_{\boldsymbol{\varphi}_{(0,a,b)}}\mathscr{L}&=\partial_{\boldsymbol{\varphi}_{(0,a,b)}}D_1\mathscr{L}\\ %\quad \text{ for } t_i\notin I\\
D_1\partial_{\boldsymbol{\varphi}_{(a,b,c)}}\mathscr{L}&=\partial_{\boldsymbol{\varphi}_{(a,b,c)}}D_1\mathscr{L}-\partial_{\boldsymbol{\varphi}_{(a-1,b,c)}}\mathscr{L} \quad \text{for }a\geq 1
\end{split}
\end{equation}
along with similar relations for $D_2$ and $D_3$.  We notice that the sum of the three terms that begin $\dfrac{\partial}{\partial \boldsymbol{\varphi}_{(l,m,n)}}$ is zero since it is the coefficient of $\delta \boldsymbol{\varphi}_{(l,m,n)}$ in $\delta \sf{dL}$.  The sum of the three terms that begin $D_1D_2D_3$ is zero by \eqref{trip}.  We can simplify further by noticing that

\begin{equation}
D_1D_2\frac{\delta \mathscr{L}_{(23)}}{\delta \boldsymbol{\varphi}_{(l,m+1,n)}}+D_1D_2\frac{\delta \mathscr{L}_{(31)}}{\delta \boldsymbol{\varphi}_{(l+1,m,n)}}=-D_1D_2\frac{\delta \mathscr{L}_{(12)}}{\delta \boldsymbol{\varphi}_{(l+1,m+1,n-1)}}
\end{equation}
\begin{equation}
D_2D_3\frac{\delta \mathscr{L}_{(31)}}{\delta \boldsymbol{\varphi}_{(l,m,n+1)}}+D_2D_3\frac{\delta \mathscr{L}_{(12)}}{\delta \boldsymbol{\varphi}_{(l,m+1,n)}}=-D_2D_3\frac{\delta \mathscr{L}_{(23)}}{\delta \boldsymbol{\varphi}_{(l-1,m+1,n+1)}}
\end{equation}
\begin{equation}
D_3D_1\frac{\delta \mathscr{L}_{(12)}}{\delta \boldsymbol{\varphi}_{(l+1,m,n)}}+D_3D_1\frac{\delta \mathscr{L}_{(23)}}{\delta \boldsymbol{\varphi}_{(l,m,n+1)}}=-D_3D_1\frac{\delta \mathscr{L}_{(12)}}{\delta \boldsymbol{\varphi}_{(l-1,m+1,n+1)}}
\end{equation}
by \eqref{trip}.  Therefore, the surviving terms form \eqref{expansion} are:

\begin{equation}
\begin{split}
&\frac{\partial \mathscr{L}_{(12)}}{\partial \boldsymbol{\varphi}_{(l,m,n-1)}}-D_1\frac{\delta \mathscr{L}_{(12)}}{\delta \boldsymbol{\varphi}_{(l+1,m,n-1)}}-D_2\frac{\delta \mathscr{L}_{(12)}}{\delta \boldsymbol{\varphi}_{(l,m+1,n-1)}}-D_1D_2\frac{\delta \mathscr{L}_{(12)}}{\delta \boldsymbol{\varphi}_{(l+1,m+1,n-1)}}\\
&\frac{\partial \mathscr{L}_{(23)}}{\partial \boldsymbol{\varphi}_{(l-1,m,n)}}-D_2\frac{\delta \mathscr{L}_{(23)}}{\delta \boldsymbol{\varphi}_{(l-1,m+1,n)}}-D_3\frac{\delta \mathscr{L}_{(23)}}{\delta \boldsymbol{\varphi}_{(l-1,m,n+1)}}-D_2D_3\frac{\delta \mathscr{L}_{(23)}}{\delta \boldsymbol{\varphi}_{(l-1,m+1,n+1)}}\\
&\frac{\partial \mathscr{L}_{(31)}}{\partial \boldsymbol{\varphi}_{(l,m-1,n)}}-D_3\frac{\delta \mathscr{L}_{(12)}}{\delta \boldsymbol{\varphi}_{(l,m-1,n+1)}}-D_1\frac{\delta \mathscr{L}_{(12)}}{\delta \boldsymbol{\varphi}_{(l+1,m-1,n)}}-D_3D_1\frac{\delta \mathscr{L}_{(31)}}{\delta \boldsymbol{\varphi}_{(l+1,m-1,n+1)}}
\end{split}
\end{equation}
which, by \eqref{split}, is precisely

\begin{equation}
\frac{\delta\mathscr{L}_{(12)}}{\delta \boldsymbol{\varphi}_{(l,m,n-1)}}+\frac{\delta\mathscr{L}_{(23)}}{\delta \boldsymbol{\varphi}_{(l-1,m,n)}}+\frac{\delta\mathscr{L}_{(31)}}{\delta \boldsymbol{\varphi}_{(l,m-1,n)}}.
\end{equation}
This is the coefficient of $\delta \boldsymbol{\varphi}_{(l,m,n)}$ and by setting this equal to zero, we have a new relation at the $M-1^{th}$ order.  We should note that if more than one of $l,m$ and $n$ is non-zero then the coefficient of $\delta \boldsymbol{\varphi}_{(l,m,n)}$ will contribute to more than one of the coefficients at the next order down.  However, since each coefficient is identically zero, the integrand is unchanged if we multiply each coefficient by 1,2 or 3 as required, corresponding to how many of $l,m$ and $n$ are non-zero.\\

\noindent Whilst this proof applies only to a Lagrangian 2-form, it is relatively straightforward to generalize this argument to get the multiform EL equations for a Lagrangian k-form.

\section{Proof of Theorem \ref{thm}}\label{ap:A}

\noindent The Lagrangian 2-form $\mathscr{L}_{(\xi\eta)}\rm d\xi\wedge d\eta+\mathscr{L}_{(\eta\nu)}d\eta\wedge d\nu+\mathscr{L}_{(\nu\xi)}d\nu\wedge d\xi$ is closed if and only if $D_\nu\mathscr{L}_{(\xi\eta)}+D_\xi\mathscr{L}_{(\eta\nu)}+D_\eta\mathscr{L}_{(\nu\xi)}=0$ on solutions of the system.

\begin{equation*}
D_\nu\mathscr{L}_{(\xi\eta)}+D_\xi\mathscr{L}_{(\eta\nu)}+D_\eta\mathscr{L}_{(\nu\xi)}
\end{equation*}
\begin{equation}\label{s1}
\begin{split}
=&\tr\bigg\{\sum_{i=1}^{N_1}[(\varphi^i)^{-1}\varphi_\eta^i,(\varphi^i)^{-1}\varphi_\nu^i]\bar U^i +\sum_{j=1}^{N_2}[(\psi^j)^{-1}\psi_\nu^j,(\psi^j)^{-1}\psi_\xi^j]\bar V^j+\sum_{k=1}^{N_3}[(\chi^k)^{-1}\chi_\xi^k,(\chi^k)^{-1}\chi_\eta^k]\bar W^k
\end{split}
\end{equation}
\begin{equation}\label{s2}
\begin{split}
&+\sum_{i=1}^{N_1}W^0_\eta U^i +W^0U^i_\eta- V^0_\nu U^i-V^0U^i_\nu\\&+\sum_{j=1}^{N_2}U^0_\nu V^j+U^0V^j_\nu-W^0_\xi V^j-W^0V^j_\xi\\&+\sum_{k=1}^{N_3}V^0_\xi W^k+V^0W^k_\xi-U^0_\eta W^k-U^0W^k_\eta
\end{split}
\end{equation}

\begin{equation}\label{s3}
-\sum_{i=1}^{N_1}\sum_{j=1}^{N_2}\frac{V^j_\nu U^i+V^jU^i_\nu}{a_i-b_j}-\sum_{j=1}^{N_2}\sum_{k=1}^{N_3}\frac{W^k_\xi V^j+W^kV^j_\xi}{b_j-c_k}-\sum_{k=1}^{N_3}\sum_{i=1}^{N_1}\frac{U^i_\eta W^k+U^iW^k_\eta}{c_k-a_i}\bigg\}.
\end{equation}

\noindent The first set of terms (part (\ref{s1})) are equivalent to

\begin{equation*}
\tr\bigg\{\sum_{i=1}^{N_1}\varphi_\eta^i(\varphi^i)^{-1} U^i_\nu+\sum_{j=1}^{N_2}\psi_\nu^j(\psi^j)^{-1} V^j_\xi+\sum_{k=1}^{N_3}\chi_\xi^k(\chi^k)^{-1} W^k_\eta\bigg\}.
\end{equation*}

\noindent We can use the compatibility conditions to re-write this as

\begin{equation*}
\tr\bigg\{\sum_{i=1}^{N_1}\sum_{k=1}^{N_3}\frac{U^i_\eta W^k}{c_k-a_i}-\sum_{i=1}^{N_1}U^i_\eta W^0
+\sum_{j=1}^{N_2}\sum_{i=1}^{N_1}\frac{V^j_\nu U^i}{a_i-b_j}-\sum_{j=1}^{N_2}V^j_\nu U^0
+\sum_{k=1}^{N_3}\sum_{j=1}^{N_2}\frac{W^k_\xi V^j}{b_j-c_k}-\sum_{k=1}^{N_3}W^k_\xi V^0\bigg\}
\end{equation*}

\noindent and we see that all of these terms will cancel with terms in parts (\ref{s2}) and (\ref{s3}).  Also, we can simplify part (\ref{s2}) by using the compatibility conditions on $U^0$, $V^0$ and $W^0$ to get that

\begin{equation*}
\begin{split}
D_\nu\mathscr{L}_{(\xi\eta)}&+D_\xi\mathscr{L}_{(\eta\nu)}+D _\eta\mathscr{L}_{(\nu\xi)}\\
=\tr\bigg\{&\sum_{i=1}^{N_1}\sum_{j=1}^{N_2}\frac{V^jU^i_\nu}{b_j-a_i}+\sum_{j=1}^{N_2}\sum_{k=1}^{N_3}\frac{W^kV^j_\xi}{c_k-b_j}+\sum_{k=1}^{N_3}\sum_{i=1}^{N_1}\frac{U^iW^k_\eta}{a_i-c_k}\\
&+\sum_{i=1}^{N_1}(-V^0U^i_\nu+[V^0,W^0]U^i)+\sum_{j=1}^{N_2}(-W^0V^j_\xi+[W^0,U^0]V^j)+\sum_{k=1}^{N_3}(-U^0W^k_\eta+[U^0,V^0]W^k)\bigg\}.
\end{split}
\end{equation*}

\noindent We use the compatibility conditions again on the terms in the first line to re-write this as

\begin{equation*}
\begin{split}
\tr\bigg\{&\sum_{i=1}^{N_1}\sum_{j=1}^{N_2}\sum_{k=1}^{N_3}V^j[U^i,W^k]\bigg(\frac{1}{(b_j-a_i)(c_k-a_i)}+\frac{1}{(c_k-b_j)(a_i-b_j)}+\frac{1}{(a_i-c_k)(b_j-c_k)}\bigg)\\
&+\sum_{i=1}^{N_1}\sum_{j=1}^{N_2}\frac{[W^0,U^i]V^j}{b_j-a_i}+\sum_{j=1}^{N_2}\sum_{k=1}^{N_3}\frac{[U^0,V^j]W^k}{c_k-b_j}+\sum_{k=1}^{N_3}\sum_{i=1}^{N_1}\frac{[V^0,W^k]U^i}{a_i-c_k}\\
&+\sum_{i=1}^{N_1}(-V^0U^i_\nu+V^0,[W^0,U^i])+\sum_{j=1}^{N_2}(-W^0V^j_\xi+W^0[U^0,V^j])+\sum_{k=1}^{N_3}(-U^0W^k_\eta+U^0[V^0,W^k])\bigg\}.
\end{split}
\end{equation*}

\noindent The first line of this is easily seen to be zero.  The remainder of this expression is then

\begin{equation*}
\begin{split}
\tr\bigg\{&-V^0\sum_{i=1}^{N_1}\bigg(U^i_\nu+[U^i,W^0+\sum_{k=1}^{N_3}\frac{W^k}{a_i-c_k}]\bigg)\\
&-W^0\sum_{j=1}^{N_2}\bigg(V^j_\xi+[V^j,U^0+\sum_{i=1}^{N_1}\frac{U^i}{b_j-a_i}]\bigg)\\
&-U^0\sum_{k=1}^{N_3}\bigg(W^k_\eta+[W^k,V^0+\sum_{j=1}^{N_2}\frac{V^j}{c_k-b_k}\bigg)\bigg\}
\end{split}
\end{equation*}

\noindent which is also zero since the summed terms are zero by the compatibility conditions.

\bibliography{ZMref}

\bibliographystyle{unsrt}

\end{document}